\documentclass[reprint,aps,pra,superscriptaddress,showkeys]{revtex4-2}
\usepackage[english]{babel}
\usepackage{siunitx}
\usepackage{physics}
\usepackage{bbm}
\usepackage{amsthm}
\usepackage{amssymb}
\usepackage{graphicx}
\usepackage{wasysym}
\usepackage{xcolor}
\usepackage[normalem]{ulem}

\newtheorem{theorem}{Theorem}[section]
\newtheorem{proposition}{Proposition}

\newtheorem{lemma}[theorem]{Lemma}

\theoremstyle{remark}

\newcommand{\id}{\mathbbm{1}}

\let\emptyset\varnothing

\begin{document}

\title{Analog neutral-atom for in-memory processing in
quantum reservoir computing 
}

\author{Luca Nigro}
\affiliation{Dipartimento di Fisica, Università degli Studi di Milano, Via Celoria 16, 20133 Milano, Italy}
\author{Gian Luca Giorgi}
\affiliation{Institute for Cross-Disciplinary Physics and Complex Systems (IFISC) UIB-CSIC, Campus Universitat Illes Balears, 07122 Palma de Mallorca, Spain}
\author{Enrico Prati}
\affiliation{Dipartimento di Fisica, Università degli Studi di Milano, Via Celoria 16, 20133 Milano, Italy}
\affiliation{Istituto di Fotonica e Nanotecnologie, Consiglio Nazionale delle Ricerche, Piazza Leonardo da Vinci 32, 20133 Milano, Italy}
\author{Roberta Zambrini}
\affiliation{Institute for Cross-Disciplinary Physics and Complex Systems (IFISC) UIB-CSIC, Campus Universitat Illes Balears, 07122 Palma de Mallorca, Spain}

\date{\today}




 


\begin{abstract}
Quantum Reservoir Computing (QRC) exploits the rich dynamics of quantum many-body systems to process time-dependent information with high-dimensional state spaces. 
While neutral atom arrays offer a scalable platform for this paradigm, realizing intrinsic temporal memory without relying on external classical buffering requires precise control over the relaxation dynamics of the system. We address the necessity of specific non-unitary dynamics for enforcing fading memory while overcoming the lack of separability.
This study simulates and empirically identifies the specific dynamical regimes required to maximize the computational capacity of the reservoir, addressing the amount of dissipation, the effectiveness of lying at the edge of quantum chaos, and that of time-multiplexing compared to the scaling of the physical size of the system.
Here we show that controlled dissipation is strictly necessary to induce echo state property, fading memory, and separability, in neutral atom arrays.
Computational performance increases at the edge of quantum chaos. The counter-intuitive residual memory provided by the unital phase damping channel is justified by introducing a model based on pure outputs. The consequent implementation enables the solution of complex non-linear tasks, such as the Mackey-Glass time series, with as few as $N=5$ atoms.
We establish a rigorous framework for engineering dissipative quantum reservoirs required for intrinsic temporal processing on near-term quantum devices.
\end{abstract} 


	
\maketitle

\section{Introduction}
Implementing machine learning on Noisy Intermediate-Scale Quantum (NISQ) hardware remains challenging~\cite{PRXQuantum.1.020101,rev.q.alg}. To circumvent the barren plateaus~\cite{babbush2018,Fontana2024,larocca2025barren} that hinder variational methods, non-variational algorithms have emerged as robust alternatives. Techniques including quantum kernel methods~\cite{schuld2021machine}, quantum principal component analysis~\cite{lloyd2014quantum}, quantum feature maps~\cite{simen2025quenchedquantumfeaturemaps}, quantum reservoir computing (QRC)~\cite{Fujii2017}, and quantum extreme learning machines (QELM)~\cite{innocenti2023potential} project complex data into high-dimensional Hilbert spaces via fixed quantum evolution. Relying solely on non-parametrized dynamics, these approaches inherently bypass standard trainability bottlenecks.

Within this non-variational framework, Reservoir Computing (RC) stands out for its efficacy in time-series analysis~\cite{jaeger2004, lukosevicius2009, tanaka2019, nakajima2021}. The distinctive feature of RC is the use of an input-driven dynamical system, called a reservoir, to non-linearly map the input data into a high-dimensional feature space. Such architectures can approximate any target continuous functional with arbitrary precision, provided they exhibit fading memory, a property formally extended to the quantum domain~\cite{Chen2020, martinez-pena2021,Nokkala2021,Sannia2024,monzani2025}. 
Following Fujii and Nakajima~\cite{Fujii2017} proposal, 
this class of algorithms has been expanded to exploit various quantum substrates~\cite{mujal2021} leading to recent experimental reservoir computing and extreme learning machines implementations~\cite{Chen2020,Suprano2024,McMahon2023,kornjaca2024,Cimini2026,Carles2026,Hou2026,Paparelle2026}. 
QRC harnesses natural quantum dynamics as a reservoir, where the non-linearity of quantum observables and many-body correlations provide a feature set whose complexity scales exponentially with the number of units. 
This high-dimensional quantum mapping offloads the computational burden from the training process, allowing the final output to be reconstructed through a simple classical linear regression layer optimized solely on the reservoir's readout.

This representational power and training efficiency have sparked significant interest, leading to theoretical advances that explore spatial multiplexing~\cite{nakajima2019}, distributed architectures~\cite{GarciaBeni2025}, dynamical phase transitions~\cite{martinez-pena2021}, measurement protocols~\cite{mujal2023,morguisancho2026}, dissipation~\cite{Sannia2024,Domingo2023, gotting2025}, non-markovianity \cite{Sannia2026}, quantum coherence~\cite{Palacios2024}, feedback mechanisms~\cite{Kobayashi2024}, topological effects \cite{Sannia2025}, and quantum probing~\cite{Kobayashi2025}. Beyond theory, this momentum has driven concrete hardware proposals~\cite{Dudas2023, llondra2025, Ricci2026, monzani2025nonunital} and practical applications spanning medical data analysis~\cite{antoncich2026}, chaotic forecasting~\cite{li2025b}, quantum circuit compression~\cite{ghosh2021}, financial prediction~\cite{vitali2025, otieno2026}, and image denoising~\cite{das2025b}.

An effective reservoir requires three foundational conditions: the echo state property (ESP), fading memory, and separability~\cite{lukosevicius2009}. ESP and fading memory ensure that the system asymptotically forgets its initial state while retaining a finite input history. Simultaneously, separability maps distinct temporal inputs to distinguishable Hilbert space regions. Lacking these intrinsic properties, architectures relying on external classical buffering reduce to QELMs~\cite{innocenti2023potential}, serving merely as static, non-linear feature extractors. True temporal quantum processing requires moving beyond QELMs, engineering systems where memory emerges directly from the quantum dynamics.

Neutral atom arrays excited to Rydberg states offer a highly suitable platform for this challenge due to their tunable geometry and strong van der Waals interactions~\cite{Saffman2010, Bernien2017, Browaeys2020, Manetsch2025}. The ability to assemble and coherently control large numbers of identical atoms has driven the rapid development of scalable commercial quantum processors~\cite{Henriet2020, Graham2022, wurtz2023}. Recent breakthroughs have begun to exploit this potential for machine learning architectures. For instance, Bravo et al.~\cite{bravo2022} introduced a quantum Recurrent Neural Network (qRNN) model using Rydberg arrays. While their framework successfully demonstrates brain-inspired cognitive functions, it primarily operates as a generic quantum neural network rather than a dissipative temporal reservoir governed by fading memory.

Transitioning specifically to temporal processing, initial efforts achieved time-series prediction on Rydberg platforms, as notably demonstrated by Kornjača et al.~\cite{kornjaca2024, beaulieu2025}. However, by relying on classically buffered, time-windowed inputs, their approach does not deploy in-memory processing and effectively operates within the QELM paradigm, utilizing the quantum array primarily as a static, non-linear feature map rather than an active dynamical memory. Following this, Settino et al.~\cite{settino2025} explored hybrid solutions that explicitly deploy classical memory to bypass the finite temporal retention of the quantum system. Most recently, Liu et al.~\cite{liu2026} advanced practical QRC implementations in Rydberg atom arrays, showing immense promise across diverse application domains~\cite{batista2025, vitali2025, antoncich2026}. Despite these rapid advancements, achieving intrinsic memory management without external classical buffering remains a critical challenge. To operate QRC, the necessity of dissipation has been investigated across various quantum platforms~\cite{Sannia2024, das2025, gotting2025}.

In this work, we bridge this critical gap by establishing the fundamental physical conditions required to operate an analog neutral atom array as a true dynamical reservoir. Rather than relying on classical buffers or local digital-analog control, we prioritize scalable global driving and continuous dissipative dynamics. We demonstrate under which open-system maps the ESP, fading memory, and input separability are intrinsically guaranteed, explicitly confirming the necessity of non-unitary, non-unital, ergodic, and mixing evolution.
Having established these theoretical requirements, we systematically optimize the physical and operational parameters of the reservoir using the linear Short-Term Memory (STM) task. This benchmark allows us to map the performance of the atomic quantum reservoir across its dynamical phases, revealing that memory capacity increases when the Hamiltonian is tuned to the edge of quantum chaos~\cite{yaakov2002, Pena2023, llondra2025, kobayashi2026}. We perform a comprehensive comparison of relevant noise channels (amplitude damping (AD), phase damping, and depolarizing) as a function of the decay rate $\gamma$. 
Notably, our analysis reveals that, while the amplitude damping is optimal, the typically detrimental phase damping channel preserves a residual temporal memory. The STM task shows that  detuning encoding outperforms Rabi amplitude encoding. Furthermore, time-multiplexing saturates quickly due to multicollinearity, making larger physical system size the more robust route to larger expressivity.

Following this structural optimization, we evaluate the reservoir's ability to process higher-order temporal correlations by testing its non-linear memory capacity at increasing polynomial degrees $p$ and in complex chaotic forecasting tasks. Finally, we tackle the highly irregular Santa Fe time series to compare our fully autonomous approach against previous methodologies~\cite{kornjaca2024}. By evaluating the impact of classical time-windowing across both paradigms, we definitively quantify the computational advantage of intrinsic quantum fading memory over static feature extraction.

The manuscript is organized as follows. In Section~\ref{sec:neutralatom_qrc}, we introduce the neutral atom QRC model. Section~\ref{sec:dissipation} characterizes the dissipative channels using the pure output formalism and proves the restoration of ergodicity for phase damping. Section~\ref{sec:results} presents the numerical performance on non-linear benchmarks, highlighting the superiority of the AD channel and the edge-of-chaos regime. Finally, Section~\ref{sec:conclusion} summarizes the design principles for quantum information processing.


\section{Neutral atoms array for QRC}\label{sec:neutralatom_qrc}

To implement a quantum reservoir, we utilize the highly tunable dynamics of neutral atom arrays trapped in optical tweezers. This platform offers natural scalability, long coherence times, and strong, controllable interactions, making it an ideal physical substrate for processing complex temporal data.
The internal dynamics of the quantum reservoir are governed by the interacting Rydberg Hamiltonian. We consider an array of $N$ atoms, each modeled as a two-level system consisting of a ground state $|g\rangle$ and a highly excited Rydberg state $|r\rangle$. The coherent evolution of the system is described by:
\begin{equation}\label{eq:hamiltonian}
  H = \frac{\Omega(t)}{2} \sum_{i=1}^N \sigma_x^{(i)} - \Delta(t) \sum_{i=1}^N n^{(i)} + \sum_{i<j} V_{ij} n^{(i)} n^{(j)}
\end{equation}
where $\Omega(t)$ is the transverse driving Rabi frequency, $\sigma_x^{(i)}$ is the Pauli-X operator coupling the ground and Rydberg states of the $i$-th atom, and $\Delta(t)$ is the global laser detuning. The operator $n^{(i)} = |r\rangle\langle r|^{(i)}$ represents the Rydberg state population at site $i$. The term $V_{ij} = C_6 / d_{ij}^6$ represents the strongly repulsive van der Waals interaction between atoms $i$ and $j$ separated by a distance $d_{ij}$. 

To process sequential data, the classical input stream $u_k$ must be continuously injected into the quantum state. The comprehensive study by Kornjaca \textit{et al.}~\cite{kornjaca2024} classifies data injection in programmable neutral atom arrays into three distinct paradigms: (i) position encoding, which modulates the effective Rydberg interaction strength by engineering the physical geometry of the atoms, (ii) local pulse encoding, which utilizes site-dependent detunings to address individual atoms, and (iii) global encoding, where a single time-varying laser profile acts simultaneously on the entire array.

Ref.~\cite{kornjaca2024} demonstrates that local pulse encoding yields the highest predictive performance, achieving site-resolved control, for example via spatial light modulators (SLMs)~\cite{Graham2022}, but introduces significant engineering overhead. In contrast, we deliberately adopt the mechanically simpler global encoding. 
Although the extended temporal schedules required for continuous global driving demand long coherence times, the approach is exceptionally hardware-efficient with an $\mathcal{O}(1)$ scaling: a single laser profile drives the entire array simultaneously, so the encoding overhead does not grow with system size. Because our primary focus lies in the simulation regime, aiming to identify fundamental physical requirements and optimal operational phases to guide future hardware, we prioritize the theoretical simplicity and robust scalability offered by the global approach.  

In our continuous-time paradigm, the time-dependent input signal directly modulates the global driving fields. At each discrete time step $k$, the input value $u_k \in [0,1]$ is held constant over a duration $\Delta t$, producing a piecewise-constant driving protocol. Specifically, the input sequence is mapped onto the global detuning such that 
\begin{equation} 
\Delta_k = \Delta(u_k) \in \left[\Delta_{\text{center}} -\delta,\, \Delta_{\text{center}}+ \delta\right], 
\end{equation} 
where $\Delta_{\text{center}}$ is the operating point around which the input modulates the detuning, and $\delta$ dictates the encoding range. We will refer to such a protocol as the global detuning encoding. Alternatively, we explore mapping the input sequence onto the transverse Rabi frequency such that 
\begin{equation} 
\Omega_k = \Omega(u_k) \in \left[\Omega_{\text{center}} -\delta_\Omega,\, \Omega_{\text{center}}+ \delta_\Omega\right], 
\end{equation}
where $\Omega_{\text{center}}$ and $\delta_\Omega$ are the corresponding base frequency and encoding range. We will refer to such a protocol as the global Rabi encoding. The optimal values of all mentioned parameters are identified in
Appendix~\ref{app:input_encoding}.

To contextualize our architecture, we contrast our approach with standard analog implementations~\cite{kornjaca2024}, which process sequential data by buffering it into a classical time window. This paradigm relies on an external memory over a finite temporal input to drive a predominantly unitary evolution, ultimately extracting multiple-time outcomes rather than continuously integrated features. As recent theory highlights~\cite{liu2026}, such window-based protocols function more akin to QELMs because their temporal retention is not intrinsic to the quantum system. In contrast, our methodology completely eliminates classical buffering to establish a QRC architecture. By explicitly optimizing open-system dissipative channels and tuning the Hamiltonian to the edge of quantum chaos, we ensure the array intrinsically supports the echo state property. This endows the system with an instantaneous, time-local memory that autonomously integrates past inputs, forcing the array to operate as a quantum reservoir.

To formally model the open-system dynamics, let $\rho$ denote the density matrix of the full $N$-atom reservoir, acting on the Hilbert space $\mathcal{H} = (\mathbb{C}^2)^{\otimes N}$. Its time evolution under the input-driven dissipative dynamics is governed by the Lindblad master equation $\dot{\rho} = \mathcal{L}(\rho)$, where
\begin{multline}
 \mathcal{L}(\rho) = -i[H(u_k), \rho] \\
+ \sum_{i=1}^{N} \sum_{\mu} \gamma_\mu 
\left( L_\mu^{(i)} \rho L_\mu^{(i)\dagger} 
- \tfrac{1}{2} \left\{ L_\mu^{(i)\dagger} L_\mu^{(i)}, \rho \right\} \right).
\label{eq:lindbladian}
\end{multline}
Here $H(u_k)$ is the input-dependent Rydberg Hamiltonian of Eq.~\eqref{eq:hamiltonian}, $\gamma_\mu > 0$ is the decay rate associated with the $\mu$-th dissipative channel, $L_{\mu}^{(i)}$ are the corresponding Lindblad jump operators acting locally on atom $i$, and $\{\cdot,\cdot\}$ denotes the anticommutator. 
At each time step $k$, the input $u_k$ enters through the global driving fields as defined in Section~\ref{sec:neutralatom_qrc}, and the reservoir state $\rho_k$ is updated by integrating this equation over the duration $\Delta t$. The readout observables are then extracted from $\rho_k$ as described below.

We evaluate four primary noise channels:
\begin{description}
    \item[Amplitude damping (AD)] Characterized by $L_{\text{AD}}^{(i)} = \sigma_-^{(i)} = \ket{g}\bra{r}^{(i)}$, modeling spontaneous emission from the Rydberg state to the ground state.
    \item[Generalized amplitude damping (GAD)] The finite-temperature extension of AD, incorporating both spontaneous emission $L_{\text{GAD},-}^{(i)} = \sqrt{f}\,\sigma_-^{(i)}$ and thermal absorption $L_{\text{GAD},+}^{(i)} = \sqrt{1-f}\,\sigma_+^{(i)}$. Here, $f \in [0,1]$ is the emission fraction, with $f=1$ recovering the pure AD channel.
    \item[Phase damping (PD)] Defined by $L_{\text{PD}}^{(i)} = \sigma_z^{(i)}$, representing pure dephasing.
    \item[Depolarizing noise (Dep)] Applies isotropic dissipation via $L_{\text{Dep}}^{(i)} \in \{\sigma_x^{(i)}, \sigma_y^{(i)}, \sigma_z^{(i)}\}$.
\end{description}
Rather than relying on a rigorous microscopic derivation for strongly interacting regimes, we adopt a phenomenological approach, assuming independent interactions with a Markovian bath for each atom. Each channel can be interpreted as the continuous weak monitoring of an atom by a weakly coupled environment, whose stochastic measurement outcomes are averaged over. Such averaging washes out the conditional back-action and recovers the deterministic Lindblad generator above.

At the end of each time step $k$, the reservoir state $\rho_k$ is measured to extract the classical feature vector $\mathbf{x}_k$. For global detuning encoding, the features consist of the single-site expectation values $\langle \sigma_z^{(i)} \rangle_k = \mathrm{Tr}[\sigma_z^{(i)} \rho_k]$, representing the Rydberg-state population of atom $i$, together with the connected two-body correlators $\langle \sigma_z^{(i)} \sigma_z^{(j)} \rangle_k = \mathrm{Tr}[\sigma_z^{(i)} \sigma_z^{(j)} \rho_k]$ for all pairs $i < j$. For global Rabi encoding, the same structure is adopted, yielding $\langle \sigma_x^{(i)} \rangle_k$ and $\langle \sigma_x^{(i)} \sigma_x^{(j)} \rangle_k$. The full feature vector $\mathbf{x}_k$ is then passed to the classical readout layer, detailed in Appendix~\ref{app:simulation}.


\section{Geometry of Dissipative Dynamics}\label{sec:dissipation}
Efficient quantum reservoir computing fundamentally relies on channels that support fading memory and the echo state property. Previous studies~\cite{Sannia2024,Pena2023} suggest that these temporal processing capabilities are best achieved through open-system evolutions that are non-unital, ergodic, and mixing. 

To ensure that our reservoir exploits genuine quantum memory, where temporal information is intrinsically preserved during a coherent evolution rather than buffered through a classical intermediate step, we restrict our focus exclusively to non-entanglement-breaking (non-EB) channels. A map $\Phi$ is entanglement-breaking if the joint state $(\Phi \otimes \mathbb{I}) \rho_{AB}$ is separable for any initial bipartite state $\rho_{AB}$~\cite{horodecki2003}. Any such channel reduces to a measure-and-prepare (M\&P) scheme~\cite{Holevo1998}, effectively acting as a classical bottleneck. By focusing on non-EB channels, we guarantee that the reservoir maintains the potential for many-body entanglement, operating in a regime fundamentally distinct from M\&P architectures~\cite{Vieira2025}.

To formally characterize these non-EB dissipative mechanisms, we adopt the pure output (PO) formalism~\cite{Braun2014}. The set of pure outputs, $\text{PO}(\Phi)$, is defined as the set of all pure states contained within the image of the quantum channel $\Phi$. Because our Lindblad master equation models dissipation as an independent, localized interaction between each atom and the Markovian bath, we apply this geometric classification at the single-qubit level. Understanding the intrinsic geometric signature of the local dissipative map on the individual Bloch sphere is a preliminary step before analyzing how the global many-body Hamiltonian collective dynamics emerges.

Applying this framework to our noise models reveals distinct geometric deformations. The depolarizing (Dep) channel represents an isotropic contraction of the Bloch sphere that maps all pure states to mixed states ($\text{PO}_{\text{Dep}} = \emptyset$). This contraction is symmetric around the origin, confirming its unital nature (it preserves the maximally mixed state $\id/2$). The phase damping (PD) channel is unital, but it suppresses coherences while preserving populations, leaving exactly two orthogonal pure states invariant ($\text{PO}_{\text{PD}} = \{|0\rangle, |1\rangle\}$). Conversely, the amplitude damping (AD) channel is explicitly non-unital. By driving the system toward a specific ground state, it possesses a single pure output ($\text{PO}_{\text{AD}} = \{|0\rangle\}$), a process that permanently displaces the center of the Bloch sphere.

We now connect this PO geometric classification to the asymptotic dynamical properties required for QRC: ergodicity and mixing. Formally, a map is ergodic if it possesses a unique fixed point $\rho_\text{ss}$, and mixing if any initial state asymptotically converges to this fixed point as $t\to\infty$~\cite{Burgarth2007}. We unify these concepts for local non-EB qubit channels generated by continuous Markovian dynamics (Lindblad semigroups) through the following proposition (rigorous proofs are provided in Appendix~\ref{app:proofs_po}):

\begin{proposition}\label{prop:po}
Given the Lindblad generator $\mathcal{L}$, let $\Phi_t = e^{\mathcal{L}t}$ be a continuous-time, non-EB, single-qubit channel. The following properties hold:
\begin{enumerate}
\item[(a)] If $\Phi_t$ admits exactly two POs, it is unital, non-ergodic, and non-mixing.
\item[(b)] If $\Phi_t$ admits a unique PO, it is ergodic and mixing toward that pure state.
\item[(c)] If $\Phi_t$ admits no POs, it is ergodic and mixing.
\end{enumerate}
\end{proposition}

Proposition~\ref{prop:po} is established at the single-qubit level and
serves as a heuristic guide for the many-body reservoir, whose validity
is confirmed numerically by the ESP and fading-memory tests in the next Section.

Based on Proposition~\ref{prop:po}, both the local AD (1 PO) and Dep (0 POs) channels inherently exhibit the ergodic and mixing dynamics desired for reservoir computing. Furthermore, strictly adhering to the criterion of non-unitality for optimal fading memory~\cite{Pena2023}, the AD channel emerges as the most naturally suitable candidate. 
Conversely, the PD channel explicitly violates ergodicity due to its two orthogonal POs ($|0\rangle, |1\rangle$). Intuitively, pure dephasing preserves the populations along the quantization axis, meaning any incoherent mixture on the $z$-axis acts as a steady state. While this theoretical limitation would seemingly disqualify PD as a reservoir resource, this conclusion assumes the absence of a driving field. The many-body Rydberg Hamiltonian introduces terms (such as the transverse drive $\Omega \sum \sigma_x^{(i)}$) that do not commute with the PD Lindblad operators ($[H, L_z] \neq 0$).

Geometrically, this coherent drive generates rotations that continuously displace the POs ($|0\rangle, |1\rangle$) away from the invariant $z$-axis. Once displaced, these states are subjected to the dissipative action of the PD channel and contracted into the interior of the Bloch sphere. Consequently, the combined open-system map effectively possesses zero pure outputs and, according to Proposition~\ref{prop:po}(c), this global symmetry breaking restores both ergodicity and mixing. We stress, however, that the resulting dynamics remains unital: as the spectral analysis of Appendix~\ref{app:restoration} shows, its unique fixed point is the maximally mixed state, which is input-independent, so that phase damping cannot sustain input separability in the asymptotic limit. The residual memory of the driven PD channel is instead a genuinely finite-time effect: the symmetry-breaking drive removes the steady-state degeneracy but leaves a small Liouvillian gap, decelerating the relaxation toward the maximally mixed state and preserving a usable trace of past inputs over intermediate timescales. A rigorous spectral analysis of the Liouvillian validating this mechanism is provided in Appendix~\ref{app:restoration}.

\section{Results}\label{sec:results}
In this Section, we present the performance of the neutral atom quantum
reservoir across a series of increasingly complex temporal processing
benchmarks. All physical and simulation parameters are detailed in
Appendix~\ref{app:simulation}.

First, we utilize the Short-Term Memory (STM) task~\cite{jaeger2001} to identify the optimal dynamical regime and characterize the critical role of different dissipative channels in establishing fading memory. We also leverage this task to evaluate the scalability of the physical array compared to time-multiplexing, a technique in which features are
extracted not only at time $k\Delta t$ but also at $V-1$ additional
sub-intervals within each step, yielding $V$ virtual nodes.
Based on these findings, we then evaluate the ability of the atomic reservoir to process complex non-linear dependencies using the NARMA-$p$ benchmark~\cite{Atiya2000narma}. Finally, we push the system to its computational limits by testing its forecasting capabilities on two distinct chaotic datasets: the Mackey-Glass time series~\cite{mackeyglass1977}, where we assess both one-step-ahead accuracy and long-term horizons, and the Santa Fe time series~\cite{weigend1994santafe}, which provides a definitive benchmark to contrast our autonomous quantum memory against classically buffered architectures. Across all subsequent benchmark tasks, we consistently partition the data set by setting the train/test threshold at $66\%$ of the total time series, ensuring a rigorous outside-of-sample evaluation of the remaining data.

\subsection{Dynamical features}

Before deploying the neutral atom array on these complex computational benchmarks, we must rigorously verify that the system satisfies the three fundamental requirements of reservoir computing: the echo state property (ESP), fading memory, and separability. Figure~\ref{fig:echo_state_property} evaluates these requisite properties across three distinct dissipative channels. 

\begin{figure}[tp]
\centering
\includegraphics[width=\linewidth]{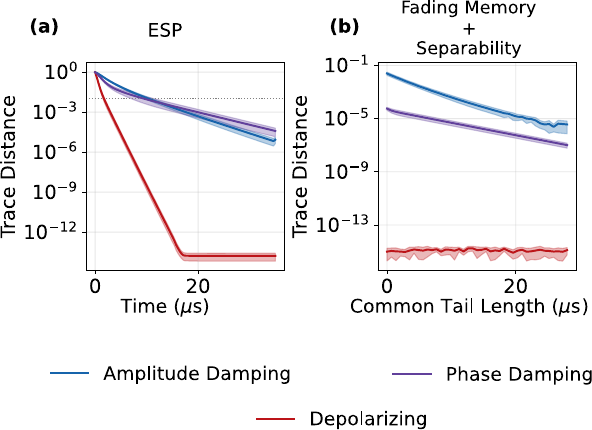}
\caption{Validation of reservoir computing properties for $N=4$ atoms
under amplitude damping (AD, blue), phase damping (PD, purple), and
depolarizing (Dep, red) channels at a fixed decay rate
$\gamma=0.4\,\mu\text{s}^{-1}$.
(a) Trace distance between reservoirs initialized in $10$ distinct
random states and driven by identical inputs. The decay to zero confirms
the echo state property (ESP): Dep converges almost instantaneously, while
AD and PD both drop below $1\%$ within $10\,\si{\micro\second}$ and below
$0.1\%$ after $20\,\si{\micro\second}$, defining the washout period.
(b) Trace distance between reservoirs driven by distinct input sequences
that share a common tail of increasing length.
Dep fails separability entirely, PD loses it over longer timescales, while AD retains distinguishability over the longest timescale, with a four-order-of-magnitude decay followed by a slowly relaxing regime.}
\label{fig:echo_state_property}
\end{figure}

To rigorously test the ESP, we generate a set of $10$ random initial quantum states. Due to the high dimensionality of the Hilbert space, any two such randomly selected states are nearly orthogonal, providing a stringent test of initial state independence. We initialize identical reservoirs with these states, drive them using the exact same input time series, and evaluate the trace distance for every possible pair. 
As shown in panel~(a) of Figure~\ref{fig:echo_state_property}, the pairwise
trace distances asymptotically approach zero for all noise models, confirming
that the reservoir successfully forgets its initial conditions.
While the Dep channel erases this initial state dependence almost instantaneously, the AD and PD channels still achieve a trace distance below $1\%$ within $10\,\si{\micro\second}$ and drop to $0.1\%$ after $20\,\si{\micro\second}$. This $20\,\si{\micro\second}$ convergence timescale physically dictates our washout period. 
Corresponding to approximately $57$ time steps in our simulated sequences, we systematically discard these initial transient states during the training phase to ensure that the readout layer only optimizes over states driven entirely by the input history rather than the initial quantum state.

Next, we assess fading memory and separability. To strictly isolate the influence of the input history, we initialize two reservoirs in the identical quantum state and drive them with $10$ pairs of distinct input sequences. Each pair of sequences differs initially but converges into an identical final segment, referred to as a common tail. As illustrated in panel~(b) of Figure~\ref{fig:echo_state_property}, fading
memory is demonstrated by the gradual decay of the trace distance as the
length of this common tail increases. 
Fading memory and separability refer to different temporal regimes. Fading memory demands that the influence of inputs preceding the common tail decays as the tail grows, with the ESP guaranteeing its eventual vanishing in the infinite-tail limit. Separability instead demands that input histories differing \textit{within the memory horizon} of the reservoir produce distinguishable states, i.e., that distinct recent inputs are not collapsed onto the same state faster than the memory timescale itself. In the common-tail experiment, the signature of a viable reservoir is therefore a trace distance that starts above the noise level and decays on a finite, task-relevant timescale. 

The Dep channel fails already at the initial stage: after the washout, both reservoirs have already relaxed to the input-independent fixed point. The PD channel exhibits the correct signature: an initial distance that decays as the tail grows. However, this occurs on a short timescale, which is consistent with the weak residual memory discussed in the previous section. The AD channel exhibits the desired behavior: the trace distance decreases by four orders of magnitude over the entire simulated period, far exceeding that of PD.
These foundational tests confirm that AD is uniquely suited to process temporal information, motivating our focus on this dissipative regime in the following performance evaluations.

\subsection{Short-term memory: linear regime}
To evaluate the temporal processing capabilities of the neutral atom array, we employ the Short-Term Memory (STM) task~\cite{jaeger2001}. This benchmark quantifies the ability of the reservoir to retain and reconstruct past inputs, providing a direct measure of its fading memory. Specifically, the STM task requires the linear readout to reconstruct a past input signal $u_{k-\tau}$ at a given delay $\tau$, where the input stream $u_k$ consists of independent and identically distributed (i.i.d.) random values sampled uniformly from the interval $[0, 1]$. The performance at each delay is evaluated using the memory capacity ($MC_\tau$), defined as the squared Pearson correlation coefficient between the target sequence and the predictions:
\begin{equation}
    MC_\tau = \frac{\text{Cov}^2\left[y_k(\tau), \widehat{y}_k\right]}{\text{Var}\left[y_k(\tau)\right]\text{Var}\left[\widehat{y}_k\right]},
\end{equation}
where $y_k(\tau) = u_{k-\tau}$ is the target and $\widehat{y}_k$ is the predicted output. To be formally consistent throughout this work, these statistical moments are evaluated as empirical temporal averages over the discrete test interval of length $L_\text{test}$. Specifically, we denote this temporal expectation as $\mathbb{E}[X] = \frac{1}{L_\text{test}} \sum_{k=1}^{L_\text{test}} X_k$, such that the variance is $\text{Var}[X] = \mathbb{E}\left[(X - \mathbb{E}[X])^2\right]$ and the covariance is $\text{Cov}[X, Y] = \mathbb{E}\left[(X - \mathbb{E}[X])(Y - \mathbb{E}[Y])\right]$.

\begin{figure*}[tp]
\centering
\includegraphics[width=0.9\linewidth]{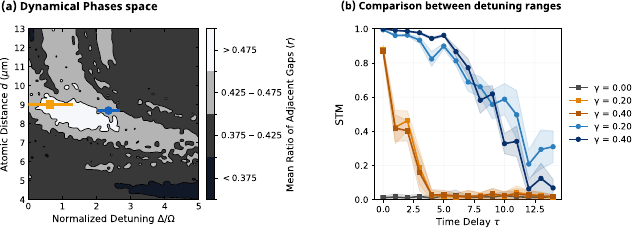}
\caption{Phase space analysis and memory capacity comparison for the Short Term Memory (STM) task for a system of 8 atoms. 
(a) Contour plot of the mean ratio of adjacent gaps characterizing the dynamical phase space as a function of the interatomic distance $d$ and the normalized global detuning $\Delta/\Omega$. The yellow region indicates the standard operating range utilized in previous QuEra implementations ($d=9\,\mu\text{m}$, $\Delta_{\text{center}}=0.64\,\Omega$)~\cite{kornjaca2024}, while the blue region highlights our proposed operating point at the edge of chaos ($d=8.69\,\mu\text{m}$, $\Delta_{\text{center}}=2.35\,\Omega$). 
(b) Memory capacity as a function of the delay time $\tau$. The colors correspond to the operating regimes defined in panel (a). For both regimes, we evaluate the unitary baseline ($\gamma=0$, gray) and dissipative dynamics ($\gamma=0.2\,\mu\text{s}^{-1}$, light shades, $\gamma=0.4\,\mu\text{s}^{-1}$, dark shades). The results demonstrate that positioning the system at the edge of chaos (blue) significantly outperforms the prior baseline (yellow) under dissipative conditions.}
\label{fig:memory_capacity_phasespace}
\end{figure*}

Before comparing different dissipative channels, we first determine the optimal physical parameters of the Rydberg array to maximize its intrinsic computational capacity.
Figure~\ref{fig:memory_capacity_phasespace}(a) illustrates the dynamical phase space of a system of 8 atoms at different interatomic distances $d$ and the normalized global detunings $\Delta/\Omega$. 
In order to characterize such dynamical phases, first we evaluate the ratio of adjacent gaps $r_n = \min(\Delta E_n, \Delta E_{n+1})/\max(\Delta E_n, \Delta E_{n+1})$, where $\Delta E_n = E_{n}-E_{n-1}$ is the gap between adjacent eigenenergies, and then take the average $\expval{r}$~\cite{Oganesyan2007}. 
This spectral statistic serves as a precise diagnostic tool to distinguish between integrable (Poisson statistics) and chaotic (Wigner-Dyson statistics) dynamics. Within this phase space, we highlight two distinct operational regimes. The yellow region represents the standard parameter space utilized in~\cite{kornjaca2024}, with $d=9\,\mu\text{m}$, $\Delta_{\text{center}}=0.64\,\Omega$, and a detuning range $\delta = 0.64\,\Omega$. In contrast, the blue region marks our proposed operating point. Here, we set $d=8.69\,\mu\text{m}$ to achieve a higher participation ratio of the all-ground initial state across the Rydberg Hamiltonian eigenbasis, hence a higher generalized multifractal dimension, and we center the detuning at $\Delta_{\text{center}}=2.35\,\Omega$ with a tighter range $\delta = 0.3\,\Omega$ 
to position the system exactly at the critical transition edge between integrability and full quantum chaos~\cite{prati2015towards}.

To validate this architectural choice, Figure~\ref{fig:memory_capacity_phasespace}(b) compares the memory capacity as a function of the delay $\tau$ for both configurations. We evaluate the performance under purely unitary and dissipative evolutions. As theoretically expected, the unitary baseline fails to exhibit fading memory. However, upon introducing controlled dissipation, our optimized configuration at the edge of chaos significantly outperforms the previous baseline, exhibiting a substantially higher total memory capacity and retaining information over longer delay steps. 
These results confirm that while dissipation is strictly necessary to enforce the echo state property, tuning the global Hamiltonian to the edge of chaos is equally crucial to maximize the high-dimensional processing power of the atomic reservoir.

With the physical array parameters secured at the edge of chaos, we next isolate the impact of different noise models based on the performance of the reservoir. The results, summarized in Figure~\ref{fig:noise_comparison}, are averaged over $10$ independent realizations of the random input signal to ensure statistical robustness. We observe that the unitary baseline exhibits negligible memory since, in the absence of dissipation, the system fails to fade its past inputs.

\begin{figure*}[tp]
\centering
\includegraphics[width=0.8\linewidth]{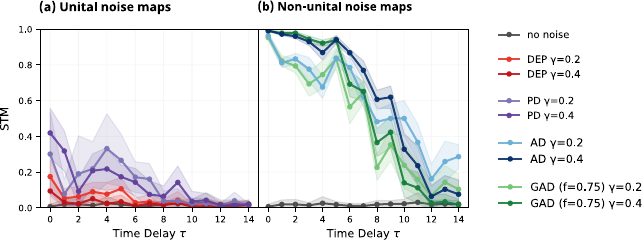}
\caption{Memory capacity on the Short-Term Memory (STM) task under different dissipative regimes. 
(a) Unital channels: depolarizing (Dep, red) and phase damping (PD, purple). 
(b) Non-unital channels: amplitude damping (AD, blue) and generalized amplitude damping (GAD, green). 
We consider here again $N=8$ atoms.
Solid lines represent the average performance over $10$ independent input realizations, while the shaded bands denote the standard deviation.
Lighter colors correspond to a dissipation rate $\gamma=0.2\,\mu\text{s}^{-1}$, while darker colors indicate $\gamma=0.4\,\mu\text{s}^{-1}$. 
The gray line represents the unitary baseline ($\gamma=0$), which fails to encode memory due to the lack of the echo state property. 
While Dep erases information by driving the system to the maximally mixed state, AD maximizes capacity by acting as an ergodic, state-purifying resource. Notably, PD exhibits non-trivial memory despite being unital, due to Hamiltonian-induced ergodicity restoration.}
\label{fig:noise_comparison}
\end{figure*}

As can be seen in Figure~\ref{fig:noise_comparison}(a), in the unital regime, the Dep channel behaves as expected, driving the system toward the maximally mixed state, effectively erasing all stored information and resulting in near-zero capacity.
However, the PD channel presents a notable anomaly. Despite being intrinsically unital, non-mixing, and non-ergodic, it exhibits a non-negligible memory capacity. Tests with a very low Rabi frequency validate our theoretical model of ergodicity restoration: the interplay between the dephasing operators and the non-commuting Rydberg Hamiltonian prevents the system from stagnating in a symmetry subspace, allowing it to retain a small amount of memory. 

As shown in Figure~\ref{fig:noise_comparison}(b), the non-unital channels yield the highest performance. Specifically, the AD channel achieves the global maximum in memory capacity: by being non-unital, ergodic, and mixing, it drives the system toward a pure state (the input dependent ground state), balancing the necessary contraction of the state space with the retention of input information. The generalized amplitude damping (GAD) follows a similar trend but performs slightly worse than AD, likely due to the shift of the steady-state towards a mixed state.

Since amplitude damping provides the optimal dissipative environment, we now tackle the fundamental challenge of scalability. Beyond increasing the physical system size, it is common, also in QRC \cite{Martinez-Pena2023}, to exploit time-multiplexing. Figure~\ref{fig:memory_capacity_atom_virtual} contrasts these two approaches by analyzing the total memory capacity ($MC_{\Sigma 15}$) summed over a total of $15$ delays:
\begin{equation}
    MC_{\Sigma 15} = \sum_{\tau=0}^{15} MC_\tau.
\end{equation}

In the case of physical scaling shown in Figure~\ref{fig:memory_capacity_atom_virtual}(a), we observe a robust, monotonic increase in memory capacity as the number of atoms $N$ grows. This stability indicates that each additional atom contributes linearly independent degrees of freedom to the state space of the atomic reservoir, confirming that the long-range Rydberg interactions effectively distribute information across the physical array without introducing redundant correlations.

Conversely, the time-multiplexing approach in Figure~\ref{fig:memory_capacity_atom_virtual}(b) demonstrates that introducing virtual nodes $V$ leads to a strong enhancement of the capacity, increasing almost linearly up to a saturation limit, which in this case is reached at $V = 5$. Beyond this optimal threshold, the capacity plateaus due to the onset of multi-collinearity within the extracted observables. As the number of virtual time steps increases, the subsequent measurement outcomes become increasingly linearly dependent, failing to span new, orthogonal dimensions in the feature space of the reservoir \cite{Martinez-Pena2023}.

Beyond optimizing the internal Hamiltonian and the dissipative environment, the performance of a quantum reservoir heavily depends on how classical data is injected into the quantum system and how the resulting high-dimensional state is extracted. Until now, we have relied on the established global detuning encoding coupled with population measurements in the computational basis ($\sigma_z$ readout). To fully explore the capabilities of the neutral atom platform, we investigate alternative strategies by varying how information is injected and extracted. We contrast longitudinal detuning encoding with transverse Rabi amplitude encoding, and pair each with either $\sigma_z$ or $\sigma_x$ measurements. The resulting four configurations also allow us to evaluate performance when the readout observables commute with the encoding operators, as well as when they do not.

\begin{figure}[tp]
\centering
\includegraphics[width=0.9\linewidth]{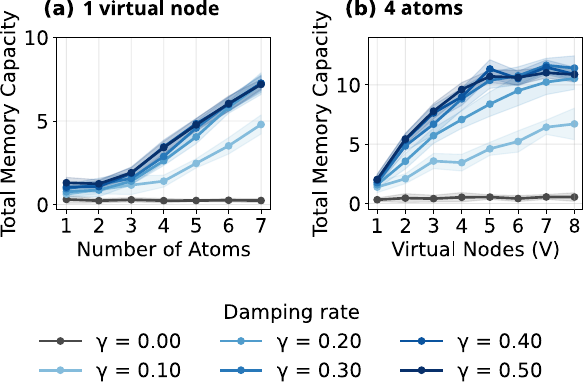}
\caption{Total memory capacity summed over delays $\tau=0$ to $15$ under amplitude damping. We compare the performance of physical scaling against time-multiplexing strategies. 
(a) Capacity as a function of the number of physical atoms $N$ with $V=1$. 
(b) Capacity as a function of the number of virtual nodes $V$ for a fixed system size of $N=4$. 
The curves correspond to damping rates $\gamma$ ranging from $0.0$ (gray, unitary) to $0.5$. 
Solid lines represent the average over $10$ independent realizations while shaded regions denote the $\pm 1\sigma$ standard deviation.}
\label{fig:memory_capacity_atom_virtual}
\end{figure}

To rigorously evaluate these configurations, we simulate an open-system quantum
register of $N=8$ atoms driven by a random input time series of $100$ timesteps following the wash-out phase.
Throughout these evaluations, the system is subjected to AD with a fixed decay rate
of $\gamma = 0.4\,\mu\text{s}^{-1}$. Figure~\ref{fig:rabi_encoding} compares the
memory capacity across four distinct configurations, testing both encoding channels
(global detuning vs.\ global Rabi) and readout bases ($\sigma_z$ vs.\ $\sigma_x$).
In all cases, single-qubit expectation values are augmented with the corresponding
all-to-all two-body correlators: $\langle \sigma_z^{(i)} \sigma_z^{(j)} \rangle$
for the $\sigma_z$ readout and $\langle \sigma_x^{(i)} \sigma_x^{(j)} \rangle$
for the $\sigma_x$ readout.

In Figure~\ref{fig:rabi_encoding}(a), we track the memory capacity as a function of the time delay $\tau$. Here, we observe a clear performance hierarchy regarding the input channel: the global detuning encoding maintains much higher memory retention across extended delays, distinctly outperforming the Rabi amplitude encoding regardless of the chosen readout. Less obvious, however, is the optimal readout strategy. 

This is clarified in Figure~\ref{fig:rabi_encoding}(b), which integrates the total memory capacity over $15$ delay steps for each strategy. The results reveal that the best-performing configuration overall is the global detuning encoding combined with the $\sigma_x$ readout, achieving a total memory capacity of $10.6$, roughly $9\%$ higher than standard population measurements ($\sigma_z$). 
Reading out transverse magnetizations therefore provides slightly better features for the linear regression layer. Regardless of the readout basis, Rabi amplitude encoding performs worst overall, with a reduction of more than $50\%$ compared to the detuning-based strategies.

\begin{figure}[tp]
\centering
\includegraphics[width=0.99\linewidth]{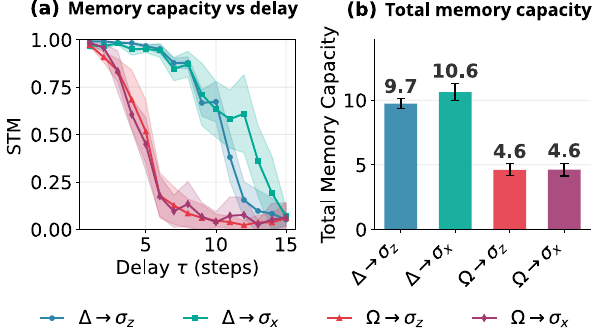}
\caption{Comparison of input encoding and readout strategies on the
Short-Term Memory (STM) task, evaluated on a system of $N=8$ atoms with
$V=4$ virtual nodes under amplitude damping at $\gamma=0.4\,\mu\text{s}^{-1}$.
(a) Memory capacity as a function of delay $\tau$ for four configurations:
global detuning encoding with $\sigma_z$ readout (blue), global detuning
encoding with $\sigma_x$ readout (green), Rabi amplitude encoding with
$\sigma_z$ readout (red), and Rabi amplitude encoding with $\sigma_x$
readout (purple). In all cases, single-qubit expectation values are
augmented with all-to-all two-body correlators to form the complete feature
vector.
(b) Total memory capacity integrated over $15$ delay steps for each
strategy. All data represent the mean over $10$ independent input
realizations, shaded bands denote $\pm1\sigma$.}
\label{fig:rabi_encoding}
\end{figure}

\subsection{NARMA-$p$: non-linear regime}

To evaluate how well the reservoir is capable of processing complex, non-linear temporal dependencies, we test its performance on the well-established NARMA-$p$ (Nonlinear Autoregressive Moving Average) benchmark~\cite{Atiya2000narma}. The task requires the reservoir to predict the output $y_{k+1}$ of a highly non-linear dynamical system driven by an independent and identically distributed random input stream $u_k$ sampled uniformly from the interval $[0, 0.5]$. Depending on the temporal order $p$, the target sequence is generated according to specific recurrence relations. For $p=2$, the dynamics are governed by
\begin{equation}
y_{k+1} = a y_k + b y_k y_{k-1} + c u_k^3 + d.
\end{equation}
For higher orders ($p > 2$), the target incorporates a non-linear moving average over the past $p$ steps alongside a delayed input product
\begin{equation}
    y_{k+1} = a y_k + b y_k \sum_{i=0}^{p-1} y_{k-i} + c u_{k-p+1} u_k + d.
\end{equation}
The specific parameters $\{a, b, c, d\}$ are selected to ensure stable but highly complex dynamics, as detailed in Figure~\ref{fig:narma}. 
To rigorously quantify the predictive accuracy, we utilize the Standardized Mean Squared Error (SMSE), defined as the mean squared error normalized by the variance of the target sequence
\begin{equation}
\text{SMSE} = \frac{\mathbb{E}\left[(y_k - \widehat{y}_k)^2 \right]}{\text{Var}(y_k)}
\end{equation}
where $y_k$ is the true target and $\widehat{y}_k$ is the prediction. An SMSE of 0 indicates perfect prediction, whereas an SMSE of order 1 indicates performance comparable to the trivial mean predictor.
For these evaluations, we fix $V=4$ to benefit from time-multiplexing without incurring multi-collinearity. Based on our characterization of the dissipative channels, we focus on the amplitude damping model and evaluate the predictive accuracy for task orders $p$ ranging from 2 to 10. The test SMSE for these tasks is presented in Figure~\ref{fig:narma}.

In the purely unitary regime ($\gamma=0$), the reservoir lacks fading memory and cannot resolve the non-linear combination of past inputs. Consequently, the system fails to fit the target dynamics, resulting in an SMSE greater than $10$ across all tested values of $p$. 
However, the introduction of amplitude damping restores the echo state property and allows the quantum reservoir to successfully capture the complex transformations of the input sequence. For the dissipative regimes, the reservoir achieves high predictive accuracy. We observe a monotonic increase in the SMSE from approximately $0.07$ for $p=2$ up to $0.4$ for $p=10$. This gradual degradation in performance is theoretically expected, as higher-order NARMA tasks require the reservoir to retain and non-linearly mix information over increasingly longer temporal windows, challenging the capacity limits of the neutral atoms array. Nevertheless, the low error rates for intermediate task orders confirm that the analog neutral atom platform is highly capable of executing non-linear time-series forecasting when properly regularized by dissipation.

Motivated by these results, a natural research question is whether the physical array can be pushed further to predict highly non-linear and chaotic dynamics. We explore this boundary in the following Section by evaluating the reservoir on the Mackey-Glass forecasting task.
\begin{figure}[tp]
\centering
\includegraphics[width=0.9\linewidth]{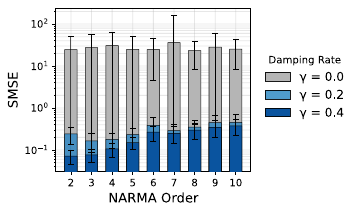}
\caption{Test Standardized Mean Squared Error (SMSE) for the NARMA-$p$ tasks for orders $p$ ranging from 2 to 10 under amplitude damping. 
Bars indicate different dissipation rates with $\gamma=0$ (gray) , $\gamma=0.2\,\mu\text{s}^{-1}$ (light blue), and $\gamma=0.4\,\mu\text{s}^{-1}$ (dark blue). 
Error bars denote one standard deviation evaluated over 100 independent random input realizations. 
The parameters $(a, b, c, d)$ for the NARMA equation are set to $(0.4,\,0.4,\,0.6,\,0.1)$ for $p=2$, and to $(0.3,\,0.05,\,1.5,\,0.1)$ for $p > 2$. 
The unitary baseline ($\gamma=0$) completely fails to fit the non-linear dynamics, yielding an SMSE $> 10$ across all tasks. Conversely, the introduction of dissipation enables accurate forecasting, with the SMSE monotonically increasing from approximately $0.07$ to $0.4$ as the temporal complexity $p$ grows.}
\label{fig:narma}
\end{figure}

\subsection{Mackey-Glass: chaotic forecasting and intrinsic memory}
The Mackey-Glass time-delay differential equation provides an ideal testbed for assessing chaotic forecasting capabilities~\cite{mackeyglass1977}. Widely recognized as a standard benchmark for reservoir computing, the system exhibits a complex chaotic attractor and an infinite-dimensional phase space, presenting a significantly harder challenge than the NARMA task.

To rigorously test the predictive power of our analog platform, we fix the number of virtual nodes to $V=6$ and first determine the optimal operating parameters for one-step-ahead prediction ($k=1$). Figure~\ref{fig:mackey-glass}(a) illustrates the test SMSE as a function of the amplitude damping rate $\gamma$. We identify a clear performance optimum at $\gamma=0.6\,\mu\text{s}^{-1}$. Conversely, operating at dissipation rates that are either too low or excessively high leads to a pronounced degradation in predictive accuracy. This optimal value represents the ideal balance for the fading memory timescale: it dissipates information fast enough to prevent the chaotic history from saturating the reservoir, yet slowly enough to retain the short-term dependencies required to reconstruct the attractor.

Next, we evaluate the impact of the system size $N$ on the predictive accuracy, shown in Figure~\ref{fig:mackey-glass}(b).
Notably, the reservoir achieves high precision even at small scale: an array of just $N=4$ atoms is sufficient to reach a test SMSE on the order of $10^{-5}$ in the ideal case, with further increases in system size yielding marginal changes in performance. 
Operating within this regime ($N=5$, $\gamma=0.6\,\mu\text{s}^{-1}$), the reservoir demonstrates remarkable accuracy. As visualized in Figure~\ref{fig:mackey-glass}(c), the predicted output closely matches the highly non-linear oscillations of the target Mackey-Glass time series for one-step-ahead forecasting.

Finally, we push the reservoir to its limits by increasing the prediction horizon $k$, requiring the system to forecast further into the future. Because the Mackey-Glass system is strictly chaotic, neighboring trajectories diverge exponentially, making long-term prediction notoriously difficult. Figure~\ref{fig:mackey-glass}(d) displays the SMSE as a function of $k$ ranging from $1$ to $30$. As theoretically expected, the error increases monotonically with the prediction horizon. Remarkably, the reservoir maintains a test SMSE strictly below a $0.1$ threshold for all horizons up to $k=20$. These results demonstrate that an optimized, fully analog quantum reservoir requiring as few as 5 physical atoms can successfully process and forecast complex chaotic dynamics.

\begin{figure}[tp]
\centering
\includegraphics[width=\linewidth]{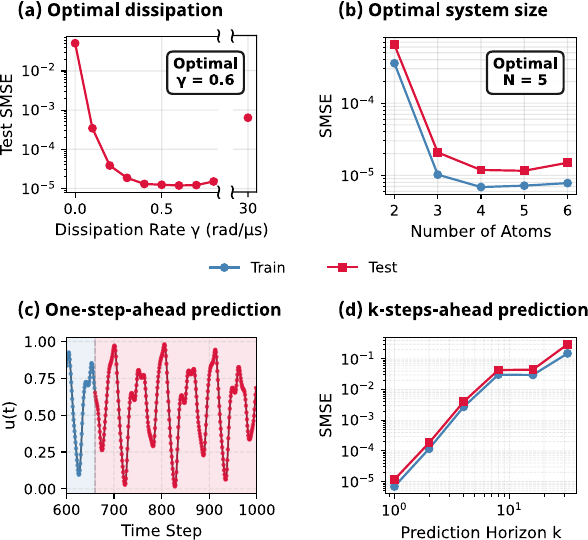}
\caption{Performance of the quantum reservoir on the chaotic Mackey-Glass forecasting task under amplitude damping ($V=6$ virtual nodes).
(a) Test SMSE for one-step-ahead prediction as a function of $\gamma$, identifying an optimal regime at $\gamma=0.6\,\mu\text{s}^{-1}$.
(b) Train (blue) and test (red) SMSE versus the number of physical atoms $N$. High predictive accuracy is achieved already at small scale, reaching a test SMSE on the order of $10^{-5}$ for $N=4$, with performance effectively saturating as $N$ increases.
(c) One-step-ahead prediction (red) overlaid on the target Mackey-Glass signal (black) for a representative array size ($N=5$, $\gamma=0.6\,\mu\text{s}^{-1}$). The two curves are nearly indistinguishable, with the target largely obscured by the prediction, indicating high forecasting accuracy.
(d) Train (blue) and test (red) SMSE as a function of the prediction horizon $k$. Performance degrades monotonically as expected for a chaotic system; the reservoir maintains a test SMSE below $0.1$ for all horizons up to $k=20$.}

\label{fig:mackey-glass}
\end{figure}

\subsection{Santa Fe: quantum memory versus classical time-windowing}

To further validate these chaotic forecasting capabilities and to provide a direct benchmark against prior state-of-the-art implementations, we evaluate the reservoir on the widely recognized Santa Fe time-series prediction task~\cite{weigend1994santafe}. To ensure robust statistical evaluation, we employ a temporal subsampling technique on the full dataset (consisting of $8000$ timesteps). Specifically, we train and test the reservoir on a smaller part of $2000$ timesteps, iterating this process $20$ times from different randomly selected starting indices. This procedure yields the average performance and standard deviations, preventing the metrics of the model from being biased by a single specific temporal trajectory.

It is crucial to emphasize that the genuine quantum memory evaluated in the previous forecasting sections corresponds strictly to a classical time-window size of $1$. In those tasks, the network relied entirely on the internal autonomous evolution of the physical reservoir to process history, using only instantaneous output values. In this Section, to systematically explore the interplay between quantum memory and classical buffering, we look at the forecasting performance of the one step Santa-Fe series for  a variable window size. As illustrated in the left schematic of Figure~\ref{fig:santafe}, we contrast two operational modes: a reset protocol, where the array is re-initialized to the ground state at the beginning of each input sequence, and a no-reset protocol, where the density matrix of the reservoir evolves uninterrupted, autonomously sustaining temporal history across steps. The reset protocol effectively mirrors the standard unitary, classically-buffered baseline established in recent literature~\cite{kornjaca2024}.

The performance plot in Figure~\ref{fig:santafe} tracks the test SMSE as a function of the classical time-window size, contrasting the effect of reset, and purely unitary dynamics against our optimized dissipative regime.
A few interesting aspects can be highlighted. First, in the absence of resets, the curves show the rather reduced effect of introducing a temporal window (buffer): the dissipative dynamics is only marginally improved, whereas the unitary dynamics remains markedly poor and inefficient even when increasing the window length. Indeed, the purely unitary predictive performance, in absence of reset does not display a fading memory at difference from the dissipative QRC where the system possesses an intrinsic fading memory that requires no external buffering. 

Second, when resets are included, the unitary dynamics benefits substantially, as the reset introduces an effective memory-shedding mechanism that allows information to be processed more efficiently than in the reset-free case; a similar, though comparatively modest, improvement is also observed in the dissipative case, where resets act as an additional memory-control mechanism. As the classical time-window size increases, the reset protocol gradually improves. Under this protocol, the quantum array acts purely as a static non-linear feature map, meaning that any temporal memory is exclusively provided by the external classical buffer. In this regime, the unitary evolution performs adequately but remains strictly bound by the computational bottleneck of the classical window size.

Finally, and consistently across all configurations, the dissipative dynamics systematically outperforms the unitary one, underscoring the role of dissipation as the dominant resource for effective information processing in this setting.
\begin{figure}[tp]
    \centering
    \includegraphics[width=\linewidth]{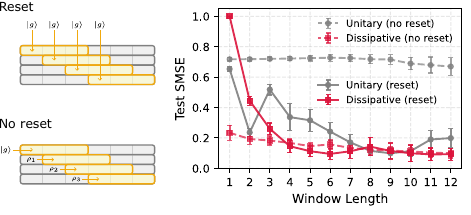}
    \caption{Performance on the chaotic Santa Fe time-series prediction task under different state-preparation protocols and classical time-window sizes. The left schematic illustrates the two operational modes: a reset protocol (solid lines), where the array is re-initialized to the ground state $|g\rangle^{\otimes N}$ at the beginning of each time window, and a no-reset protocol (dashed lines), where the final density matrix of the previous evolution step serves as the initial state for the next. The right plot displays the test Standardized Mean Squared Error (SMSE) as a function of the classical time-window size, comparing purely unitary evolution against our optimized dissipative configuration. All plotted data represent the average over $20$ independent iterations utilizing $2000$-timestep subsampled windows from a full $8000$-timestep series, with error bars indicating the standard deviation.}
    \label{fig:santafe}
\end{figure}
To quantify this advantage, the dissipative case without reset achieves an average SMSE equal to $0.23$, with standard deviation of $0.05$,  without any classical buffering (window length equal to 1). In contrast, the unitary reset protocol of \cite{kornjaca2024} with the same physical array size requires a classical time-window of size at least 7 to reach the same predicting performance.
When evaluated strictly at window length equal to 1, the unitary architecture struggles with significantly higher errors of $0.654\pm 0.015$ (reset) and $0.718\pm 0.012$ (no reset), underscoring that its performance relies almost entirely on the external buffer rather than intrinsic quantum dynamics.

By plotting these configurations side-by-side, the results clearly demonstrate how integrating controlled dissipation effectively substitutes the heavy reliance on classical time-windowing, establishing a true, autonomous quantum reservoir.


\section{Conclusion}\label{sec:conclusion}
In this work, we have comprehensively evaluated the capabilities and fundamental physical requirements of analog quantum reservoir computing using neutral atom arrays. By bridging the gap between quantum non-equilibrium dynamics and machine learning theory, we established a rigorous framework for optimizing both the physical hardware and the operational protocols of the reservoir.

A central result of our investigation is the validation of the echo state property (ESP) within the specific constraints of neutral atom arrays. While the fundamental necessity of dissipation for fading memory is well-established in the broader context of open quantum systems, this requirement has often been overlooked in the recent literature regarding neutral atom temporal processing. Most existing Rydberg-based approaches rely on purely unitary evolution, often compensating for the lack of intrinsic fading memory with classical pre-processing or time-windowing. Our findings demonstrate that an active dissipative channel is a structural path to truly quantum fading memory, bypassing the need for artificial classical buffering.

Leveraging these dissipative requirements, we utilized the universal principle of the edge of quantum chaos to optimize the operational regime of the atomic reservoir. In our architecture, we establish a linear geometry ($d=8.69\,\mu\text{m}$) and the optimal global detuning range ($\Delta_{\text{center}}=2.35\,\Omega$) that positions the Rydberg dynamics at the critical transition between integrable and quantum chaotic motion. Our results confirm that this spectral boundary provides the ideal computational landscape for neutral atom arrays.

Furthermore, our analysis of different noise channels revealed a highly non-trivial physical phenomenon regarding phase damping.
Unital channels typically degrade temporal processing: depolarizing noise
isotropically contracts the state space toward the maximally mixed state,
erasing all stored information. Phase damping, however, retains a non-negligible quantum memory.
We justify this through a theoretical model of ergodicity restoration. As proven in Proposition~\ref{prop:po} and Appendix~\ref{app:restoration}, the interplay between the dephasing operators and the non-commuting Rydberg Hamiltonian decelerates the relaxation into the maximally mixed state, allowing the system to preserve a residual memory of past inputs.


Other operational aspects have been addressed. We addressed the saturation performance for time-multiplexing and the scaling with the system size, where increasing the physical number of atoms seamlessly adds linearly independent degrees of freedom. Furthermore, our analysis of the input injection mechanism revealed a crucial operational asymmetry: encoding the temporal data into the global detuning yields drastically superior memory capacity compared to driving the Rabi amplitude.

Finally, we demonstrated that even highly compact arrays, requiring as few as $5$ physical atoms, are exceptionally capable of solving complex, real-world problems. The successful forecasting of the highly non-linear NARMA tasks, the chaotic Mackey-Glass attractor, and the Santa Fe dataset demonstrates the practical utility of the neutral atom platform. 

An important aspect of this design it that it deploys coherent in-memory processing. The superior performance in forecasting with respect to external buffering and reset strategies has also been shown displaying the importance of coherent fading memory. An alternative strategy beyond output windowing is to implement a feedback mechanism for incoherent memory~\cite{Kobayashi2024,Paparelle2026,selimovic2025} or a hybrid quantum–classical reservoir computing design as proposed with a neutral atom platform in~\cite{gyurik2025}.

By providing a rigorous physical framework for temporal information processing, this systematic approach successfully establishes the design principles for realizing analog quantum reservoir computing on neutral atoms from first principles. Ultimately, this work establishes crucial design aspects that can guide the future development of quantum reservoir computing in atomic platforms endowed with genuine, internal quantum memory. As experimental hardware continues to scale, we anticipate that these fundamental blueprints, ranging from the exploitation of the edge of chaos to the strategic harnessing of open-system dissipation, will pave the way for fully autonomous quantum processors capable of tackling complex, real-time temporal tasks natively at the quantum level.


\begin{acknowledgments}
L.N. and E.P. gratefully acknowledge financial support from Eni SpA through the co-funding of the DM 117/2023 PNRR PhD grant.
    We acknowledge support from the Spanish State Research Agency, through the María de Maeztu project CEX2021-001164-M, funded by MICIU/AEI/10.13039/501100011033; through the CoQuSy project PID2022-140506NB-C21 and -C22 funded by MICIU/AEI/10.13039/50110001103 and by ERDF, EU; and through the QuantCom project CNS2024-154720, funded by MICIU/AEI/10.13039/501100011033 and co-funded by the European Union; the project is funded under the Quantera II program that has received funding from the EU’s H2020 research and innovation program under Grant Agreement No. 101017733, and from the Spanish State Research Agency (project CoQuaDis PCI2024-153446) funded by MICIU/AEI/10.13039/50110001103; we also acknowledge CSIC's Quantum Technologies Platform (QTEP).

\end{acknowledgments}

\section*{COMPETING INTERESTS}
The authors declare no competing interests.

\section*{CODE AVAILABILITY}
The code that supports the findings of this study is available from the corresponding author upon reasonable request.


\appendix

\section{Simulation Methods}\label{app:simulation}

The interaction coefficient $C_6$ depends directly on the principal quantum number of the chosen Rydberg state $|r\rangle$. Throughout this work, we fix the Rydberg level to $n=70$, serving as a highly representative benchmark that falls within the standard operational range ($n \approx 50-100$) of current programmable neutral atom hardware, such as the Pasqal quantum processors~\cite{Henriet2020}. Operating at this physically realistic $n=70$ level yields an interaction coefficient of $C_6 \approx 5.42 \times 10^6 \, \mu\text{m}^6/\mu\text{s}$. Regarding the physical geometry, commonly referred to as the quantum register, we arrange the atoms in an open one-dimensional topology.

The extracted observables from the reservoir are concatenated into a high-dimensional state vector $\mathbf{x}_k$, which serves as the input for the classical readout layer. Over a given temporal sequence, these state vectors form the feature matrix $\mathbf{X}$, and the corresponding target outputs form the vector $\mathbf{Y}$. The ultimate prediction of the reservoir is obtained through a linear combination of these measured features, yielding the prediction $\mathbf{\widehat{Y}} = \mathbf{X} \mathbf{W}_\text{out}$. To compute the optimal readout weight vector $\mathbf{W}_\text{out}$, we employ Ridge regression:
\begin{equation}
\mathbf{W}_\text{out} = (\mathbf{X}^T \mathbf{X} + \alpha \id)^{-1} \mathbf{X}^T \mathbf{Y}.
\end{equation}
This $L_2$-regularized linear regression is strictly necessary for mitigating multi-collinearity and preventing the model from overfitting to the highly correlated training data. Throughout our evaluations, we fix the Ridge regularization parameter to $\alpha = 10^{-7}$, providing a stable and robust mapping from the quantum observables to the target output sequences.

To rigorously model these continuous-time dynamics and faithfully replicate the experimental capabilities of neutral-atom devices, all quantum evolutions in this study are executed using Pulser~\cite{Silverio2022}. By integrating the sequence design of Pulser with the QuTiP~\cite{lambert2025qutip} master-equation solver within a custom Python environment, we seamlessly construct the analog laser pulses $\Omega(t)$ and $\Delta(t)$ while computing the continuous temporal evolution of the open-system density matrix under the defined dissipative noise models. The classical readout layer, including data processing and Ridge regression, is optimized using the NumPy library~\cite{harris2020}.

\section{Input encoding configurations}\label{app:input_encoding}
The specific parameter values for both encoding protocols are as follows.
For global detuning encoding, the Rabi frequency is held fixed at
$\Omega = 3\pi\,\text{rad}/\mu\text{s}$, while the optimal
operating point and encoding range are set to $\Delta_{\text{center}} =
2.35\,\Omega$ and $\delta = 0.3\,\text{rad}/\mu\text{s}$. For
global Rabi encoding, the detuning is fixed at $\Delta =
7\pi\,\text{rad}/\mu\text{s}$, while the base frequency and encoding range are
$\Omega_{\text{center}} = 0.9\,\Delta$ and $\delta_\Omega
= 0.1\,\Delta$. In both cases, $u_k \in [0,1]$ is
linearly rescaled onto the respective modulation range, and each time step
has duration $\Delta t = 350\,\text{ns}$.

\section{Proofs of Dynamical Propositions}\label{app:proofs_po}

We provide here formal proofs of Proposition ~\ref{prop:po}, relating the pure output (PO) count to ergodicity and mixing of continuous, non-EB maps. Consider the Lindbladian
\begin{equation}
\mathcal{L}(\rho)=-i[H,\rho]+\sum_k\Big(L_k\rho L_k^\dagger
-\tfrac12\{L_k^\dagger L_k,\rho\}\Big),
\end{equation}
i.e. the GKSL generator of the evolution map $\Phi_t = e^{\mathcal{L}t}$. 
The generator splits into a relaxing (no-jump) term  $\mathcal{L}_R(\rho)=-i\big(H_{\mathrm{eff}}\rho-\rho H_{\mathrm{eff}}^\dagger\big)$, and a jump term $\mathcal{L}_J(\rho)=\sum_k L_k\rho L_k^\dagger$, where $H_{\mathrm{eff}}=H-\tfrac{i}{2}\sum_k L_k^\dagger L_k$ is the effective (non-hermitian) Hamiltonian.

\subsection{Geometric framework}
Following the notation of \cite{Baumgartner2008}, the two terms admit a sharp geometric reading on the state space $\mathbf{S}$, whose boundary $\partial\mathbf{S}$ collects the rank-deficient states (for a qubit, the pure states on the Bloch sphere). For a unit vector $|\varphi\rangle$ the functional $g_\varphi(\rho)=\langle\varphi|\rho|\varphi\rangle\ge0$ defines a supporting hyperplane of $\mathbf{S}$ at every $\rho$ with $|\varphi\rangle\in\ker\rho$; its derivative along the flow,
\begin{equation}
\label{eq:nflux}
n_\varphi(\rho):=\frac{d}{dt}\,g_\varphi(\rho(t))
=\langle\varphi|\mathcal{L}(\rho)|\varphi\rangle \ge 0
\end{equation}
is the rate at which population enters the empty direction $|\varphi\rangle$, non-negative since the flow cannot leave $\mathbf{S}$.
Only the jump term feeds it: for $|\varphi\rangle\in\ker\rho$ each term of $\mathcal{L}_R(\rho)=-i(H_{\mathrm{eff}}\rho-\rho H_{\mathrm{eff}}^\dagger)$ carries a factor $\rho|\varphi\rangle=0$ or $\langle\varphi|\rho=0$, so $\langle\varphi|\mathcal{L}_R(\rho)|\varphi\rangle=0$, whereas $\mathcal{L}_J$ injects population through the quantum jumps.
The split $\mathcal{L}=\mathcal{L}_R+\mathcal{L}_J$ is thus the separation between the drift tangent to $\partial\mathbf{S}$ and the flux normal to it. The two parts are tied by trace preservation, the anti-Hermitian identity
\begin{equation}
\label{eq:antiherm}
H_{\mathrm{eff}}-H_{\mathrm{eff}}^\dagger=-i\sum_k L_k^\dagger L_k ,
\end{equation}
which equates the no-jump depletion rate to minus the total jump rate.

\begin{lemma}[dark-state condition]\label{lemma:dark}
Let $P=|\psi\rangle\langle\psi|$. Then $\mathcal{L}(P)=0$ if and only if $P$ is a common eigenoperator of the jump superoperator $\mathcal{L}_J$ and of the relaxing superoperator $\mathcal{L}_R$. One can also check that the two eigenvalues are then opposite,
\begin{equation}
\mathcal{L}_J(P)=\Lambda\,P,\quad \mathcal{L}_R(P)=-\Lambda\,P
\end{equation}
with $\Lambda\ge 0$.
\end{lemma}
Hence, stationarity is the exact balance of the jump gain and the no-jump loss.

\textit{Proof.} ($\Rightarrow$) If $\mathcal{L}(P)=0$, the normal flux in the kernel direction $\varphi=\psi^\perp$ vanishes; by \eqref{eq:nflux} only the jump term contributes,
\begin{equation}
n_{\psi^\perp}(P)=\langle\psi^\perp|\mathcal{L}_J(P)|\psi^\perp\rangle
=\sum_k\big|\langle\psi^\perp|L_k|\psi\rangle\big|^2=0 ,
\end{equation}
a sum of squares, forcing $L_k|\psi\rangle=\lambda_k|\psi\rangle$ for every $k$. Hence $\mathcal{L}_J(P)=\Lambda P$ with $\Lambda=\sum_k|\lambda_k|^2\ge0$ and $\mathcal{L}_R(P)=-\mathcal{L}_J(P)=-\Lambda P$. ($\Leftarrow$) 
Conversely, $\mathcal{L}_J(P)\propto P$ requires $L_k|\psi\rangle=\lambda_k|\psi\rangle$, and $\mathcal{L}_R(P)\propto P$ requires $H_{\mathrm{eff}}|\psi\rangle=\mu|\psi\rangle$, giving $\mathcal{L}_R(P)=2\,\mathrm{Im}(\mu)\,P$. The former fixes $\mathrm{Im}\,\mu=\mathrm{Im}\langle\psi|H_{\mathrm{eff}}|\psi\rangle=-\Lambda/2$, so the two eigenvalues are opposite and $\mathcal{L}(P)=(\Lambda-\Lambda)P=0$.
$\square$

\begin{lemma}[orthogonality of pure steady states]\label{lemma:ortho}
A genuinely dissipative (non-unitary) qubit GKSL semigroup cannot possess two distinct, non-orthogonal pure steady states. 
\end{lemma}
Consequently, any two distinct pure steady states are orthogonal (antipodal on the Bloch sphere). Geometrically, they are two boundary points of vanishing normal flux, and the lemma forbids these tangency points from lying any closer than antipodal.

\textit{Proof.} By Lemma~\ref{lemma:dark} each steady state is such that
\begin{gather}
L_k|\psi_j\rangle=\lambda_k^{(j)}|\psi_j\rangle,\\
H_{\mathrm{eff}}|\psi_j\rangle=\mu_j|\psi_j\rangle,
\end{gather}
with $j=1,2$, $\mathrm{Im}\,\mu_j=-\tfrac12\Lambda_j$ and $\Lambda_j=\sum_k|\lambda_k^{(j)}|^2$. Sandwiching the trace-preservation identity \eqref{eq:antiherm} between the two steady rays carries this balance onto the coherence. Its off-diagonal element between $\langle\psi_1|$ and $|\psi_2\rangle$ reads
\begin{equation}
\label{eq:offrel}
(\mu_2-\bar\mu_1)\,\langle\psi_1|\psi_2\rangle
=-i\,\Lambda_{12}\,\langle\psi_1|\psi_2\rangle ,
\end{equation}
with $\Lambda_{12}:=\sum_k\overline{\lambda_k^{(1)}}\,\lambda_k^{(2)}$. Assuming $\langle\psi_1|\psi_2\rangle\neq0$, we have $\mu_2-\bar\mu_1=-i\Lambda_{12}$. Considering its imaginary part we can then write the off-diagonal balance $\Lambda_1+\Lambda_2=2\,\mathrm{Re}\,\Lambda_{12}$ which holds iff
\begin{equation}
\label{eq:key}
\sum_k\big|\lambda_k^{(1)}-\lambda_k^{(2)}\big|^2=0
\end{equation}
This means that assuming non-orthogonal pure steady states, $\lambda_k^{(1)}=\lambda_k^{(2)}=\lambda_k$ for every $k$. Two independent eigenvectors sharing every eigenvalue force $L_k=\lambda_k\id$ to be proportional to the identity. The dissipator vanishes and $\mathcal{L}(\cdot)=-i[H,\cdot]$ is purely unitary, contradicting genuine dissipation. Therefore $\langle\psi_1|\psi_2\rangle=0$. $\square$

\subsection{Proof of Proposition~\ref{prop:po}}
We use throughout the identification of $\mathrm{PO}(\Phi_t)$ with the pure steady states: the images $\Phi_t(\mathbf{S})$ form a nested family of compact convex sets, $\Phi_{t+s}(\mathbf{S})\subseteq\Phi_t(\mathbf{S})$, so a pure state that survives in every image must be fixed by the semigroup. The non-EB hypothesis is what excludes pure outputs that are not
fixed points, see the closing remark. In the language above, a PO is a boundary point of vanishing normal flux, where the flow is tangent to $\partial\mathbf{S}$.
The three cases of Proposition~\ref{prop:po} are the three ways these tangency points can populate the Bloch sphere: a pair, a single one, or none.

\paragraph{Two POs.}
By Lemma~\ref{lemma:ortho} the two POs, $P_1$ and $P_2$, are antipodal. Linearity of $\mathcal{L}$ makes the whole segment $pP_1+(1-p)P_2$ stationary, a diameter of fixed states spanning the Bloch sphere. Its midpoint is the maximally mixed state $\tfrac12(P_1+P_2)=\id/2$, so $\Phi_t$ fixes the center and is unital. The diameter of steady states is a one-parameter family, so the fixed point is not unique ($\Phi_t$ non-ergodic) and distinct inputs relax to distinct points along it ($\Phi_t$ non-mixing). $\square$

\paragraph{A single PO.}
Let $P_\psi$ be the unique PO, hence a fixed point. The orthogonal ray $\mathrm{span}\{\psi^\perp\}$ is not stationary, otherwise $P_{\psi^\perp}$ would be a second PO. Population there drains, $\langle\psi^\perp|\rho(t)|\psi^\perp\rangle\to0$ for every $\rho(0)$. 
It follows immediately that the coherence also goes to zero, $|\langle\psi|\rho(t)|\psi^\perp\rangle|^2\rightarrow 0 $
when $t\to\infty$, so $\rho(t)\to P_\psi$ for every initial state. The PO is the unique steady state (ergodic) and attracts every trajectory (mixing). $\square$

\paragraph{No PO.}
With no tangency point, the image lies strictly inside the ball, so $\Phi_t$ maps the compact convex $\mathbf{S}$ into its interior and Brouwer's theorem
gives an interior, full-rank (mixed) fixed point. It is unique: two distinct steady states would span a whole line of fixed states (by linearity of $\mathcal{L}$), and that line meets the surface of the Bloch sphere at two pure steady states, contradicting $\mathrm{PO}(\Phi_t)=\emptyset$. With a single, full-rank fixed point and the image contracting into the interior, every trajectory converges to it, so $\Phi_t$ is ergodic and mixing. $\square$

\paragraph{Remark.}
Without non-EB, pure outputs need not be steady states. The M\&P map $\Phi(\rho)=\langle0|\rho|0\rangle\,|{+}\rangle\langle{+}| +\langle1|\rho|1\rangle\,|{-}\rangle\langle{-}|$ has two orthogonal POs, $|{+}\rangle$ and $|{-}\rangle$, which are not fixed points: $\Phi(|\pm\rangle\langle\pm|)=\id/2$, and every input relaxes to $\id/2$. Readout and preparation occur in different bases, so a pure output retains no memory of its preparation. This entanglement-breaking structure is what underlies the identification of POs with steady states used throughout.
Without continuous time, Lemma~\ref{lemma:ortho} fails. Non-degenerate extremal qubit channels~\cite{Braun2014,Ruskai2002} send the Bloch ball to a displaced ellipsoid touching the sphere at two non-orthogonal pure outputs, exactly the configuration Lemma~\ref{lemma:ortho} excludes for a semigroup. Two non-orthogonal contact points are incompatible with a centered ellipsoid, so the channel is non-unital. Furthermore, it drives every state to a single mixed fixed point, hence is ergodic and mixing with two POs, against all three conclusions of statement~(a).

\section{Spectral Analysis of Ergodicity Restoration}\label{app:restoration}

The non-ergodicity of pure phase damping (PD) admits a structural reading in terms of a \emph{strong symmetry} of the Liouvillian~\cite{bucaprosen2012, albert2014}. The dephasing jump operator $L=\sigma_z$ and any commuting bare Hamiltonian (e.g., $H \propto \sigma_z$) both commute with $\sigma_z$. Consequently, $\sigma_z$ generates a strong symmetry, the observable $\langle\sigma_z\rangle$ is conserved, and the Liouvillian $\mathcal{L}$ acquires a degenerate manifold of steady states. Geometrically, this corresponds to the invariant $z$-axis of fixed points that underlies the two antipodal pure outputs $\{|0\rangle,|1\rangle\}$ described in Proposition~\ref{prop:po}(a).

Driving the array with the transverse Rydberg term $\Omega\sum_i\sigma_x^{(i)}$ explicitly breaks this strong symmetry ($[\sigma_x,\sigma_z]\neq0$), removing the conservation law and collapsing the stationary subspace to a unique steady state. We analytically demonstrate this restoration of mixing by analyzing the spectrum of the combined open-system evolution.

Expanding the density matrix in the basis $\{\id/\sqrt{2}, F_x, F_y, F_z\}$ of normalized Pauli operators, explicitly defined as $F_i = \sigma_i / \sqrt{2}$ (for $i=x,y,z$) such that $\text{Tr}(F_i F_j) = \delta_{ij}$, we write:
\begin{equation}
\rho = \frac{\id}{2} + \vec{v}\cdot\vec{F}
\end{equation}
where $\vec{v}$ is the coherence vector, a generalized analog of the Bloch vector. The Lindblad master equation $\dot{\rho} = \mathcal{L}\rho$ induces an affine transformation on $\vec{v}$:
\begin{equation}
    \dot{\vec{v}} = G\vec{v} + \vec{c}
\end{equation}
where the dynamical matrix $G$ captures the interplay between unitary rotation and dissipation. Considering a PD channel with decay rate $\gamma$ and a generic single-qubit Hamiltonian $H = a\sigma_x + b\sigma_z$, the matrix $G$ and the translation vector $\vec{c}$ take the form:
\begin{equation}
    G = \begin{pmatrix}
    -\gamma & -2b & 0 \\
    2b & -\gamma & -2a \\
    0 & 2a & 0
    \end{pmatrix}, \quad \vec{c} = 0.
\end{equation}

For purely PD ($a=b=0$) or a symmetry-preserving Hamiltonian ($a=0$), the matrix $G$ possesses a zero eigenvalue ($\lambda=0$), formally indicating the conserved subspace and resulting in non-ergodicity. However, if $a \neq 0$ (such as when the symmetry-breaking transverse drive is applied), the characteristic polynomial of $G$ yields three eigenvalues $\lambda_{1,2,3}$ that all possess strictly negative real parts. This spectral shift ensures that $\lim_{t\to\infty} \vec{v}(t) = 0$, implying an asymptotic convergence to the maximally mixed state $\rho_\text{ss} = \id/2$ regardless of the initial state. Thus, the coherent transverse driving strictly restores the mixing property.

\bibliography{bibliography}
\bibliographystyle{naturemag}

\end{document}